\documentclass[11pt,a4paper]{article}
\usepackage{amsfonts}
\usepackage{amsmath}
\usepackage{a4wide}
\usepackage{amssymb}
\usepackage{graphicx}

\newtheorem{theorem}{Theorem}

\newtheorem{lemma}{Lemma}

\newtheorem{proposition}{Proposition}

\newenvironment{proof}[1][Proof]{\textbf{#1.} }{\  \rule{0.5em}{0.5em}}
\makeatletter
\def \@removefromreset#1#2{\let \@tempb \@elt
     \def \@tempa#1{@&#1}\expandafter \let \csname @*#1*\endcsname \@tempa
     \def \@elt##1{\expandafter \ifx \csname @*##1*\endcsname \@tempa \else
    \noexpand \@elt{##1}\fi}     \expandafter \edef \csname cl@#2\endcsname{\csname cl@#2\endcsname}     \let \@elt \@tempb
     \expandafter \let \csname @*#1*\endcsname \@undefined}

\@removefromreset{equation}{section}

\@removefromreset{theorem}{section}
\makeatother
\input{tcilatex}
\begin{document}

\title{Full Bell locality of a noisy state for $N\geq 3$ nonlocally
entangled qudits }
\author{Elena R. Loubenets \\
%EndAName
National Research University Higher School of Economics, \\
Moscow, 101000, Russia}
\maketitle

\begin{abstract}
Bounds, expressed in terms of $d$ and $N,$ on \emph{full Bell locality} of a
quantum state for $N\geq 3$ nonlocally entangled qudits (of a dimension $%
d\geq 2$) mixed with white noise are known, to our knowledge, only within
full separability of this noisy $N$-qudit state. For the maximal violation
of general Bell inequalities by an $N$-partite quantum state, we specify the
analytical upper bound expressed in terms of dilation characteristics of
this state, and this allows us to find new general bounds in $d,N,$ valid
for all $d\geq 2$ and all $N\geq 3,$ on \emph{full Bell locality }under
generalized quantum measurements\emph{\ }of (i) the $N$-qudit GHZ\ state
mixed with white noise and (ii) an arbitrary $N$-qudit state mixed with
white noise. The new full Bell locality bounds are beyond the known ranges
for full separability of these noisy $N$-qudit states.
\end{abstract}

\section{Introduction}

Quantum nonlocality is now used in many quantum information processing tasks
and though, in more than 50 years since the seminal papers \cite{1, 2} of
Bell, there is still no a unique conceptual view\footnote{%
On conceptual and quantitative issues of Bell's nonlocality see the recent
article \cite{4}\ in Foundations of Physics and references therein.} on this
notion, it is nowadays clear that quantum nonlocality does not mean
propagation of interaction faster than light and is not \cite{5} equivalent
to quantum entanglement. Moreover, in quantum information, nonlocality of a
multipartite quantum state is defined purely mathematically -- \emph{via} 
\emph{violation by this state of a Bell inequality},\emph{\ }and it is
specifically in this context quantum nonlocality is now used in experimental
tasks and is discussed in the present article.

In applications, one, however, deals with noisy channels and, for a nonlocal 
$N$-partite quantum state, it is important to evaluate amounts of noise
breaking the nonclassical character of its statistical correlations. Note
that \emph{full Bell locality} of an $N$-partite quantum state, in the sense
of its nonviolation of Bell inequalities of any type and for arbitrary
numbers of settings and outcomes per site, is \emph{equivalent} (Proposition
6 in section VI of \cite{6}) to the existence of a local hidden variable
(LHV)\ model for \emph{each} correlation scenario on this state. However, as
we stressed in section 5 of \cite{3}, the latter does not necessarily imply
the existence for all scenarios on this state of a single LHV model, that
is, existence for an $N$-partite state of the LHV model formulated in \cite%
{5}.

Furthermore, one can be also interested in nonviolation by an $N$-partite
state of only some specific class of Bell inequalities, for example, Bell
inequalities for up to some specific numbers $S_{1},...,S_{N}$ of
measurement settings at $N$ sites. The latter type of \emph{partial Bell
locality} of an $N$-partite quantum state, \emph{the} $S_{1}\times \cdots
\times S_{N}$\emph{-setting Bell locality}, was analyzed in a general
setting in \cite{3, 7, 8, 8.1, 6}.

In the present paper, we analyze bounds on full Bell locality of an $N$%
-qudit state $\rho _{d,N}$ mixed with white noise: 
\begin{equation}
\beta \rho _{d,N}+\left( 1-\beta \right) \frac{\mathbb{I}_{d}^{\otimes N}}{%
d^{N}},\text{ \ \ \ }0\leq \beta \leq 1.  \label{1}
\end{equation}

Full separability of an $N$-partite quantum state implies its full locality
and for $N=2$ bounds in terms of a qudit dimension $d\geq 2$ on separability
of a noisy state (\ref{1}) were presented in \cite{9, 10}: (i) for the
two-qudit Greenberger-Horne-Zeilinger (GHZ) state $\rho _{d,2}^{(ghz)},$\ a
noisy state (\ref{1}) is separable \cite{9} if and only if $\beta \leq \beta
_{sep}^{(ghz,d,2)}=\frac{1}{d+1};$\ (ii) for an arbitrary two-qudit state $%
\rho _{d,2},$\ a noisy state (\ref{1}) is separable for all \cite{10} $\beta
\leq \beta _{sep}^{(\rho _{d,2})}$, where $\beta _{sep}^{(\rho _{d,2})}$\
varies in the range $\frac{1}{d^{2}-1}\leq \beta _{sep}^{(\rho _{d,2})}\leq 
\frac{2}{d^{2}+2}.$

For $N\geq 3$, bounds in $d,$ $N$ on full separability of a noisy $N$-qudit
state (\ref{1}) were analyzed in \cite{10, 11, 12, 13, 14, 15, 16, 17, 17.1,
17.2, 17.3} and it was found that, for an arbitrary $N$-qubit state $\rho
_{2,N}$, a noisy state (\ref{1}) is fully separable for all \cite{13} 
\begin{equation}
\beta \leq \beta _{sep}^{(\rho _{2,N})}=\frac{1}{1+2^{2N-1}}  \label{1.1}
\end{equation}%
and, for the $N$-qubit GHZ state $\rho _{2,N}^{(ghz)},$ a noisy state (\ref%
{1}) is fully separable if and only if \cite{14, 16} 
\begin{equation}
\beta \leq \beta _{sep}^{(ghz,2,N)}=\frac{1}{1+2^{N-1}}.  \label{1.2}
\end{equation}

For higher qudit dimensions $d\geq 3,$ it is now known that, for an
arbitrary $N$-qudit state $\rho _{d,N},\ $a noisy state (\ref{1}) is fully
separable if \cite{17.2} 
\begin{equation}
\beta \leq \beta _{sep}^{(\rho _{d,N})}=\frac{1}{1+d^{2N-1}},  \label{1.3}
\end{equation}%
and that there exist $N$-qudit states $\rho _{d,N},$ for which a mixed state
(\ref{1}) is fully nonseparable \cite{17.1, 17.3} for all $\beta >\frac{1}{%
1+d^{N-1}}.$ The latter is, in particular, the case for a noisy $N$-qudit
GHZ\ state (\ref{1}) -- it is fully nonseparable if \cite{17.2, 17.1, 17.3} 
\begin{equation}
\beta >\frac{1}{1+d^{N-1}}.  \label{1.3_1}
\end{equation}

Therefore, in view of (\ref{1.3}), (\ref{1.3_1}), for all $d\geq 3,$ $N\geq
3,$ the $N$-qudit GHZ\ state $\rho _{d,N}^{(ghz)}$ mixed with white noise is
fully separable for all $\beta \leq \beta _{sep}^{(ghz,d,N)},$ where the
value $\beta _{sep}^{(ghz,d,N)},$ $d\geq 3,$ $N\geq 3,$ admits the bounds 
\begin{equation}
\frac{1}{1+d^{2N-1}}\leq \beta _{sep}^{(ghz,d,N)}\leq \frac{1}{1+d^{N-1}}.
\label{1.4}
\end{equation}%
It was also proved \cite{17.1} that, for prime $d\geq 2,$ the value $\beta
_{sep}^{(ghz,d,N)}|_{prime\text{ }d}=\frac{1}{1+d^{N-1}}$ and condition $%
\beta \leq \frac{1}{1+d^{N-1}}$ is necessary and sufficient for full
separability a noisy $N$-qudit GHZ\ state (\ref{1}) with prime $d.$

\emph{Beyond full separability, }bounds in $d,$ $N$ for full Bell locality
of a noisy $N$-qudit state (\ref{1}) were studied, to our knowledge, only in
the two-qudit case, see \cite{18, 19, 20, 21, 22, 23} and references
therein. For $N\geq 3$, the important analytical and numerical results on
Bell locality of a noisy $N$-qudit state (\ref{1})\ were analysed in many
papers but in the sense of \emph{partial Bell locality}, see \cite{24, 25,
26, 27, 28, 29, 30, 31, 32} and references therein.

In the present paper, we analyze bounds on full Bell locality of a noisy $N$%
-qudit state (\ref{1}) via the LqHV (local quasi hidden variable)
mathematical formalism, introduced and developed in \cite{6, 33, 34, 34.1}.
This allows us to derive general bounds in $d,N$, valid for all $d\geq 2$
and all $N\geq 3,$ on \emph{full Bell locality }under generalized quantum
measurements of (i) the $N$-qudit GHZ\ state $\rho _{d,N}^{(ghz)}$ mixed
with white noise and (ii) an arbitrary $N$-qudit state $\rho _{d,N}$ mixed
with white noise. The new full Bell locality bounds are beyond the known
ranges (\ref{1.3}), (\ref{1.4}) for full separability of these noisy $N$%
-qudit states.

As we discuss above, to our knowledge, for arbitrary $d\geq 2,N\geq 3,$
bounds in $d,$ $N$ on full Bell locality of a noisy $N$-qudit state (\ref{1}%
) are known in the literature only within its full separability.

Note that our mathematical techniques is valid for all $d\geq 2,$ $N\geq 2$.
However, in this paper, we do not intend to reproduce or improve via this
techniques the well known bounds (see in \cite{21, 23}) for full Bell
locality of a noisy state (\ref{1}) in the two-qudit case $(N=2,d\geq 2)$.
Our main aim is to find general bounds on full Bell locality of noisy $N$%
-qudit states (\ref{1}) valid for all $d\geq 2,$ $N\geq 3$ and to study
their asymptotics for large $N$ and $d$.

The paper is organized as follows.

In Section 2, we recall the notion of a general\footnote{%
That is, a Bell inequality of any type, either on correlation functions or
on joint probabilities or of a more complicated form, for details, see the
general framework \cite{35} for multipartite Bell inequalities.} Bell
inequality and introduce \cite{6} the parameters specifying for an $N$%
-partite quantum state the maximal violation of $S_{1}\times \cdots \times
S_{N}$-setting general Bell inequalities and the maximal violation of all
general Bell inequalities.

In Section 3, for the maximal violation of general Bell inequalities by an $%
N $-partite quantum state, we present the analytical upper bound quantifying
Bell nonlocality of an $N$-partite quantum state in terms of its dilation
characteristics and this allows us to introduce a general condition (Theorem
1) sufficient for full Bell locality of an $N$-partite state under
generalized quantum measurements.

In Section 4, we apply Theorem 1 for finding new general bounds on full Bell
locality of (i) the $N$-qudit GHZ state mixed with white noise and (ii) an
arbitrary $N$-qudit state mixed with white noise and study asymptotics of
these new bounds for large $N$ and $d.$

In Section 5, we discuss the derived results.

\section{Preliminaries: quantum violation of general Bell inequalities}

In this section, we shortly recall \cite{35} the notion of a general Bell
inequality and specify the parameters \cite{6} defining the maximal
violation by an $N$-partite quantum state of (i) $S_{1}\times \cdots \times
S_{N}$\emph{-}setting\emph{\ }general\emph{\ }Bell inequalities for an
arbitrary number of outcomes at each site and (ii) $\emph{all}$ general Bell
inequalities.

This allows us to quantify analytically the $S_{1}\times \cdots \times S_{N}$%
-setting Bell locality and full Bell locality of an $N$-partite quantum
state.

Consider\footnote{%
For the general framework on the probabilistic description of an arbitrary
correlation scenario, see \cite{3}.} a general $N$-partite correlation
scenario where each $n$-th of $N\geq 2$ parties performs $S_{n}\geq 1$
measurements with outcomes $\lambda _{n}\in \Lambda _{n}$ of an arbitrary
spectral type. We label each measurement at $n$-th site by a positive
integer $s_{n}=1,...,S_{n}$ and each $N$-partite joint measurement, induced
by this correlation scenario and with outcomes 
\begin{equation}
(\lambda _{1},\ldots ,\lambda _{N})\in \Lambda =\Lambda _{1}\times \cdots
\times \Lambda _{N},  \label{2}
\end{equation}%
by an $N$-tuple $(s_{1},...,s_{N}),$ where $n$-th component specifies a
measurement at $n$-th site.

We denote by $\mathcal{E}_{S,\Lambda },$ $S=S_{1}\times \cdots \times S_{N},$
a correlation scenario with $S_{n}$ settings and outcomes $\lambda _{n}\in
\Lambda _{n}$ at each $n$-th site and by $P_{(s_{1},...,s_{N})}^{(\mathcal{E}%
_{S,\Lambda })}$ -- a joint probability distribution of outcomes $(\lambda
_{1},\ldots ,\lambda _{N})\in \Lambda $ for an $N$-partite joint measurement 
$(s_{1},...,s_{N})$ induced by this scenario.

For a general correlation scenario $\mathcal{E}_{S,\Lambda },$ consider a
linear combination 
\begin{eqnarray}
\mathcal{B}_{\Phi _{S,\Lambda }}^{(\mathcal{E}_{S,\Lambda })}
&=&\sum_{s_{1},...,s_{_{N}}}\left\langle f_{(s_{1},...,s_{N})}(\lambda
_{1},\ldots ,\lambda _{N})\right\rangle _{\mathcal{E}_{S,\Lambda }},
\label{3} \\
\Phi _{S,\Lambda } &=&\{f_{(s_{1},...,s_{N})}:\Lambda \rightarrow \mathbb{R}%
\mid s_{n}=1,...,S_{n},\text{ \ }n=1,...,N\},  \notag
\end{eqnarray}%
of averages (expectations) of the most general form:%
\begin{eqnarray}
&&\left\langle f_{(s_{1},...,s_{N})}(\lambda _{1},\ldots ,\lambda
_{N})\right\rangle _{\mathcal{E}_{S,\Lambda }}  \label{4} \\
&=&\int\limits_{\Lambda }f_{(s_{1},...,s_{N})}(\lambda _{1},\ldots ,\lambda
_{N})P_{(s_{1},...,s_{N})}^{(\mathcal{E}_{S,\Lambda })}\left( \mathrm{d}%
\lambda _{1}\times \cdots \times \mathrm{d}\lambda _{N}\right) ,  \notag
\end{eqnarray}%
specified for each joint measurement $(s_{1},...,s_{N})$ by a bounded
real-valued function $f_{(s_{1},...,s_{N})}(\cdot )$ of outcomes $\left(
\lambda _{1},\ldots ,\lambda _{N}\right) \in \Lambda $ at all $N$ sites.

Depending on a choice of a function $f_{(s_{1},...,s_{N})}$ for a joint
measurement $(s_{1},...,s_{N})$, an average (\ref{4}) may refer either to
the joint probability of events observed under this joint measurement at $%
M\leq N$ sites or, in case of real-valued outcomes, for example, to the
expectation 
\begin{equation}
{\Large \langle }\lambda _{1}^{(s_{1})}\cdot \ldots \cdot \lambda
_{n_{M}}^{(s_{n_{M}})}{\Large \rangle }_{\mathcal{E}_{S,\Lambda
}}=\int\limits_{\Lambda }\lambda _{1}\cdot \ldots \cdot \lambda
_{n_{M}}P_{(s_{1},...,s_{N})}^{(\mathcal{E}_{S,\Lambda })}\left( \mathrm{d}%
\lambda _{1}\times \cdots \times \mathrm{d}\lambda _{N}\right)  \label{5}
\end{equation}%
of the product of outcomes observed at $M\leq N$ sites or may have a more
complicated form. In quantum information, the product expectation (\ref{5})
is referred to as a correlation function. For $M=N,$ a correlation function
is called full.

Let the probabilistic description of a correlation scenario $\mathcal{E}%
_{S,\Lambda }$ admit\footnote{%
For the general statements on the LHV\ modelling, see section 4 in \cite{3}.}
\emph{a LHV (local hidden variable) model, }that is,\emph{\ }all its joint
probability distributions%
\begin{equation}
\left \{ P_{(s_{1},...,s_{N})}^{(\mathcal{E}_{S,\Lambda })},\text{ }%
s_{n}=1,...,S_{n},\text{ }n=1,...,N\right \}  \label{6}
\end{equation}%
admit the representation%
\begin{eqnarray}
&&P_{(s_{1},...,s_{N})}^{(\mathcal{E}_{S,\Lambda })}\left( \mathrm{d}\lambda
_{1}\times \cdots \times \mathrm{d}\lambda _{N}\right)  \label{7} \\
&=&\dint \limits_{\Omega }P_{1,s_{1}}(\mathrm{d}\lambda _{1}|\omega )\cdot
\ldots \cdot P_{N,s_{_{N}}}(\mathrm{d}\lambda _{N}|\omega )\nu _{\mathcal{E}%
_{S,\Lambda }}(\mathrm{d}\omega )  \notag
\end{eqnarray}%
via a single probability distribution $\nu _{\mathcal{E}_{S,\Lambda }}(%
\mathrm{d}\omega )$ of some variables $\omega \in \Omega $ and conditional
probability distributions $P_{n,s_{n}}(\mathrm{\cdot }|\omega ),$ referred
to as "local" in the sense that each $P_{n,s_{n}}(\mathrm{\cdot }|\omega )$
at $n$-th site depends only on the corresponding measurement $%
s_{n}=1,...,S_{n}$ at this site.

Then a linear combination (\ref{3}) of its averages (\ref{4}) satisfies the
tight LHV constraints (see Theorem 1 in \cite{35}): 
\begin{equation}
\mathcal{B}_{\Phi _{S,\Lambda }}^{\inf }\leq \mathcal{B}_{\Phi _{S,\Lambda
}}^{(\mathcal{E}_{S,\Lambda })}{\Large |}_{_{lhv}}\leq \mathcal{B}_{\Phi
_{S,\Lambda }}^{\sup }  \label{8}
\end{equation}%
with the LHV constants%
\begin{eqnarray}
\mathcal{B}_{\Phi _{S,\Lambda }}^{\sup } &=&\sup_{\lambda _{n}^{(s_{n})}\in
\Lambda _{n},\forall s_{n},\forall n}\text{ }%
\sum_{s_{1},...,s_{_{N}}}f_{(s_{1},...,s_{N})}(\lambda _{1}^{(s_{1})},\ldots
,\lambda _{N}^{(s_{N})}),  \label{9} \\
\mathcal{B}_{\Phi _{S,\Lambda }}^{\inf } &=&\inf_{\lambda _{n}^{(s_{n})}\in
\Lambda _{n},\forall s_{n},\forall n}\text{ }%
\sum_{s_{1},...,s_{_{N}}}f_{(s_{1},...,s_{N})}(\lambda _{1}^{(s_{1})},\ldots
,\lambda _{N}^{(s_{N})}).  \notag
\end{eqnarray}%
From (\ref{8}), it follows that, in the LHV\ case, 
\begin{equation}
\left\vert \mathcal{B}_{\Phi _{S,\Lambda }}^{(\mathcal{E}_{S,\Lambda })}%
{\Large |}_{_{lhv}}\right\vert \leq \mathcal{B}_{\Phi _{S,\Lambda
}}^{lhv}=\max \left\{ \left\vert \mathcal{B}_{\Phi _{S,\Lambda }}^{\sup
}\right\vert ,\left\vert \mathcal{B}_{\Phi _{S,\Lambda }}^{\inf }\right\vert
\right\} .  \label{10}
\end{equation}

Some the LHV inequalities in (\ref{8}) may be fulfilled for a wider (than
LHV) class of correlation scenarios. This is, for example, the case for the
LHV\ constraints on joint probabilities following explicitly from
nonsignaling\footnote{%
On this general notion, see section 3 in \cite{3}.} of probability
distributions. Moreover, some of the LHV inequalities in (\ref{8}) may be
simply trivial, i. e. fulfilled for all correlation scenarios, not
necessarily nonsignaling.

\emph{Each of the tight LHV inequalities in (\ref{8}) that may be violated
under a non-LHV scenario is referred to as a Bell (or Bell-type) inequality. 
}

Let, under an $S_{1}\times \cdots \times S_{N}$-setting correlation
scenario, each $N$-partite joint measurement $(s_{1},...,s_{N})$ be
performed on a quantum state $\rho $ on a complex Hilbert space $\mathcal{H}%
_{1}\otimes \cdots \otimes \mathcal{H}_{N}$ and be described by the joint
probability distribution 
\begin{equation}
\mathrm{tr}[\rho \{ \mathrm{M}_{1,s_{1}}(\mathrm{d}\lambda _{1})\otimes
\cdots \otimes \mathrm{M}_{N,s_{N}}(\mathrm{d}\lambda _{N})\}],  \label{11}
\end{equation}%
where each $\mathrm{M}_{n,s_{n}}(\mathrm{d}\lambda _{n})$ is a normalized
positive operator-valued (\emph{POV}) measure, representing on a complex
Hilbert space $\mathcal{H}_{n}$ a generalized quantum measurement $s_{n}$ at 
$n$-th site. For a POV measure $\mathrm{M}_{n,s_{n}}$, all its values $%
\mathrm{M}_{n,s_{n}}(F_{n}),$ $F_{n}\subseteq \Lambda _{n},$ are positive
operators on $\mathcal{H}_{n}$ and $\mathrm{M}_{n,s_{n}}(\Lambda _{n})=%
\mathbb{I}_{\mathcal{H}_{n}}.$ For concreteness, we specify this $%
S_{1}\times \cdots \times S_{N}$-setting quantum correlation scenario by
symbol $\mathcal{E}_{\mathrm{M}_{S,\Lambda }}^{(\rho )}$, where 
\begin{eqnarray}
\mathrm{M}_{S,\Lambda } &=&\left \{ \mathrm{M}_{n,s_{n}},\text{ }%
s_{n}=1,...,S_{n},\text{ }n=1,...,N\right \} ,  \label{12} \\
S &=&S_{1}\times \cdots \times S_{N},\text{ \ \ }\Lambda =\Lambda _{1}\times
\cdots \times \Lambda _{N},  \notag
\end{eqnarray}%
is a collection of POV measures at all $N$-sites.

It is well known \cite{1} that the probabilistic description of a quantum
correlation scenario $\mathcal{E}_{\mathrm{M}_{S,\Lambda }}^{(\rho )}$ does
not need \ to admit a LHV model. Therefore, under correlation scenarios $%
\mathcal{E}_{\mathrm{M}_{S,\Lambda }}^{(\rho )}$ on an $N$-partite quantum
state $\rho ,$ Bell inequalities in (\ref{8}) may be violated and, in view
of (\ref{10}) the parameter \cite{6} 
\begin{equation}
\mathrm{\Upsilon }_{S_{1}\times \cdots \times S_{N}}^{(\rho
)}=\sup_{_{\Lambda ,\text{ }\Phi _{S,\Lambda },\text{ }\mathrm{M}_{S,\Lambda
}}}\frac{1}{\mathcal{B}_{\Phi _{S,\Lambda }}^{lhv}}\left \vert \mathcal{B}%
_{\Phi _{S,\Lambda }}^{(\mathcal{E}_{\mathrm{M}_{S,\Lambda }}^{(\rho
)})}\right \vert \geq 1  \label{13}
\end{equation}%
specifies the maximal violation by an $N$-partite state $\rho $ of all $%
S_{1}\times \cdots \times S_{N}$-setting general Bell inequalities while the
parameter \cite{6}%
\begin{equation}
\mathrm{\Upsilon }_{\rho }=\sup_{S_{1},...,S_{N}}\mathrm{\Upsilon }%
_{S_{1}\times \cdots \times S_{N}}^{(\rho )}\geq 1  \label{14}
\end{equation}%
-- the maximal violation of \emph{all} general Bell inequalities.

Clearly, \emph{an }$N$\emph{-partite quantum state }$\rho $\emph{\ is the} $%
S_{1}\times \cdots \times S_{N}$-\emph{setting Bell local iff} 
\begin{equation}
\mathrm{\Upsilon }_{S_{1}\times \cdots \times S_{N}}^{(\rho )}=1  \label{15}
\end{equation}%
\emph{and} \emph{fully Bell local iff }%
\begin{equation}
\mathrm{\Upsilon }_{\rho }=1.  \label{16}
\end{equation}

\section{Quantifying Bell nonlocality}

In this section, we present the analytical upper bound on the maximal Bell
violation parameters (\ref{13}), (\ref{14}). In view of (\ref{15}), (\ref{16}%
), this allows us to specify full Bell locality of an $N$-partite quantum
state via its dilation characteristics (see Theorem 1 below).

Recall that, according to Proposition 1 in \cite{8} for a bipartite case and
Proposition 1 in \cite{6} for an arbitrary $N$-partite case, \emph{for every
state }$\rho $ on a complex Hilbert space $\mathcal{H}_{1}\otimes \cdots
\otimes \mathcal{H}_{N}$ \emph{and arbitrary integers} $S_{1},...,S_{N}\geq
1,$ \emph{there exists} \emph{an }$S_{1}\times \cdots \times S_{N}$\emph{%
-setting} \emph{source operator} $T_{S_{1}\times \cdots \times S_{N}}^{(\rho
)}$ -- that is, a self-adjoint trace class operator on the Hilbert space 
\begin{equation}
\left( \mathcal{H}_{1}\right) ^{S_{1}}\otimes \cdots \otimes \left( \mathcal{%
H}_{N}\right) ^{S_{N}}  \label{17}
\end{equation}%
satisfying the relation%
\begin{align}
& \mathrm{tr}\left[ T_{_{S_{1}\times \cdots \times S_{N}}}^{(\rho )}\left \{ 
\mathbb{I}_{\mathcal{H}_{1}^{\otimes k_{1}}}\otimes X_{1}\otimes \mathbb{I}_{%
\mathcal{H}_{1}^{\otimes (S_{1}-1-k_{1})}}\otimes \cdots \otimes \mathbb{I}_{%
\mathcal{H}_{N}^{\otimes k_{N}}}\otimes X_{N}\otimes \mathbb{I}_{\mathcal{H}%
_{1}^{\otimes (S_{N}-1-k_{N})}}\right \} \right]  \label{18} \\
& =\mathrm{tr}\left[ \rho \left \{ X_{1}\otimes \cdots \otimes X_{N}\right
\} \right] ,  \notag \\
k_{1}& =0,...,(S_{1}-1),...,k_{N}=0,...,(S_{N}-1),  \notag
\end{align}%
for all bounded linear operators $X_{1},...,X_{N}$ on Hilbert spaces $%
\mathcal{H}_{1},....,\mathcal{H}_{N},$ respectively. Here, we set $\mathbb{I}%
_{\mathcal{H}_{n}^{\otimes k}}\otimes X_{n}\mid _{_{k=0}}$ $=X_{n}\otimes 
\mathbb{I}_{\mathcal{H}_{n}^{\otimes k}}\mid _{_{k=0}}$ $=X_{n}.$

Clearly, $T_{_{1\times \cdots \times 1}}^{(\rho )}\equiv \rho $ and \textrm{%
tr}$[T_{_{S_{1}\times \cdots \times S_{N}}}^{(\rho )}]=1.$

By definition (\ref{18}), an $S_{1}\times \cdots \times S_{N}$-setting
source operator $T_{S_{1}\times \cdots \times S_{N}}^{(\rho )}$ constitutes 
\emph{a self-adjoint trace class dilation} \emph{of a state }$\rho $\emph{\ }%
on\emph{\ }$\mathcal{H}_{1}\otimes \cdots \otimes \mathcal{H}_{N}$ to the
Hilbert space (\ref{17}).

Note that, in general, a source operator does not need to be either positive
or invariant with respect to permutations of spaces $\mathcal{H}_{n}$ in $%
\left( \mathcal{H}_{n}\right) ^{S_{n}}$, see in the proof of Proposition 1
in \cite{8}. Therefore, the notion of a symmetric $(S_{1},S_{2})$ extension,
introduced for a bipartite state in \cite{7}, constitutes an $S_{1}\times
S_{2}$-setting source operator of a particular type -- positive and
symmetric. For \emph{every} $N$-partite state $\rho $ and \emph{arbitrary
integers }$S_{1},...,S_{N}\geq 1,$ a symmetric $(S_{1},...,S_{N})$ extension
does not need to exist while an $S_{1}\times \cdots \times S_{N}$-setting
source operator does always exist.

Due to the analytical upper bound (53) proved in Lemma 3 of \cite{6}, we
have the following general statement quantifying Bell nonlocality of an $N$%
-partite quantum state in terms of its dilation characteristics.

\begin{proposition}
Under generalized N-partite joint quantum measurements (\ref{11}), the
maximal Bell violation parameters (\ref{13}), (\ref{14}) are upper bounded by%
\begin{eqnarray}
\mathrm{\Upsilon }_{S_{1}\times \cdots \times S_{N}}^{(\rho )} &\leq
&\inf_{T_{S_{1}\times \cdots \times \underset{\overset{\uparrow }{n}}{1}%
\times \cdots \times S_{N}}^{(\rho )},\text{ }\forall n}\text{ }%
||T_{S_{1}\times \cdots \times \underset{\overset{\uparrow }{n}}{1}\times
\cdots \times S_{N}}^{(\rho )}||_{cov},  \label{19} \\
\mathrm{\Upsilon }_{\rho } &\leq &\sup_{S_{1},...,S_{N}}\text{ }%
\inf_{T_{S_{1}\times \cdots \times \underset{\overset{\uparrow }{n}}{1}%
\times \cdots \times S_{N}}^{(\rho )},\text{ }\forall n}\text{ }%
||T_{S_{1}\times \cdots \times \underset{\overset{\uparrow }{n}}{1}\times
\cdots \times S_{N}}^{(\rho )}||_{cov},  \notag
\end{eqnarray}%
where infimum is taken over all source operators $T_{S_{1}\times \cdots
\times \underset{\overset{\uparrow }{n}}{1}\times \cdots \times
S_{N}}^{(\rho )}$ with only one setting at some $n$-th site and over all
sites $n=1,...,N$ and notation $\left \Vert T\right \Vert _{cov}$ means the
covering norm of a self-adjoint trace class operator $T$ on space (\ref{17})
-- a new type of a norm introduced by relation (11) in \cite{6} for
self-adjoint trace class operators on an arbitrary complex Hilbert space $%
\mathcal{G}_{1}\otimes \mathcal{\cdots }\otimes \mathcal{G}_{m}$.
\end{proposition}

Recall, that, by Lemma 1 in \cite{6}, for every self-adjoint trace class
operator $W$ on $\mathcal{G}_{1}\otimes \mathcal{\cdots }\otimes \mathcal{G}%
_{m}$, the covering norm $\left \Vert W\right \Vert _{cov}$ satisfies the
relation 
\begin{equation}
\left \vert \mathrm{tr}\left[ W\right] \right \vert \leq \left \Vert W\right
\Vert _{cov}\leq \left \Vert W\right \Vert _{1},  \label{21}
\end{equation}%
where $\left \Vert \cdot \right \Vert _{1}$ is the trace norm. The relation $%
\left \Vert W\right \Vert _{cov}=\left \vert \mathrm{tr}\left[ W\right]
\right \vert $ is fulfilled if a self-adjoint trace class operator $W$ is
tensor positive (see the general definition 2 in \cite{6}), that is,
satisfies the relation\footnote{%
For a finite dimensional Hilbert space $\mathcal{G}_{1}\otimes \mathcal{G}%
_{2},$ our notion of tensor positivity is similar by its meaning to
"block-positivity" introduced in \cite{36} for a bipartite case. We,
however, consider that, for a tensor product of \emph{any number of\
arbitrary }Hilbert spaces, \emph{possibly infinite dimensional}, our term
"tensor positivity" is more suitable.}%
\begin{equation}
\mathrm{tr}\left[ W\{X_{1}\otimes \cdots \otimes X_{m}\} \right] \geq 0
\label{22}
\end{equation}%
for all positive bounded operators $X_{j}$ on $\mathcal{G}_{j},$ $j=1,...,m$.

Every positive operator on $\mathcal{G}_{1}\otimes \mathcal{\cdots }\otimes 
\mathcal{G}_{m}$ is tensor positive but not vice versa. For example, the
permutation (flip) operator $V_{d}(\psi _{1}\otimes \psi _{2}):=\psi
_{2}\otimes \psi _{1},$ $\psi _{1},\psi _{2}\in \mathbb{C}^{d},$ on $\mathbb{%
C}^{d}\otimes \mathbb{C}^{d}$ is tensor positive but is not positive. Its
trace norm is $\left \Vert V_{d}\right \Vert _{1}=d^{2}$ while the covering
norm $\left \Vert V_{d}\right \Vert _{cov}=d.$

Note that the notion of tensor positivity \cite{6} for a bounded linear
operator on a Hilbert space $\mathcal{G}_{1}\otimes \mathcal{\cdots }\otimes 
\mathcal{G}_{m}$, which reminds the notion of positivity on an arbitrary
Hilbert space, is more general than the concept of an entanglement witness\
used in quantum information for determining entanglement of a state on $%
\mathcal{G}_{1}\otimes \mathcal{\cdots }\otimes \mathcal{G}_{m}$. Namely, an
entanglement witness constitutes a tensor positive self-adjoint bounded
linear operator on $\mathcal{G}_{1}\otimes \mathcal{\cdots }\otimes \mathcal{%
G}_{m}$ which is not positive.

For each source operator $T_{S_{1}\times \cdots \times S_{N}}^{(\rho )}$,
its trace $\mathrm{tr}[T_{S_{1}\times \cdots \times S_{N}}^{(\rho )}]=1.$
Therefore, by (\ref{21}), $||T_{S_{1}\times \cdots \times S_{N}}^{(\rho
)}||_{cov}\geq 1$ and is equal to one $||T_{S_{1}\times \cdots \times
S_{N}}^{(\sigma )}$ $||_{cov}=1$ if a source operator $T_{S_{1}\times \cdots
\times S_{N}}^{(\rho )}$ is tensor positive. This and relations (\ref{13}), (%
\ref{19}), (\ref{21}) imply the following general statement \cite{6}.

\begin{lemma}
If, for an $N$-partite quantum state $\rho $ on $\mathcal{H}_{1}\otimes
\cdots \otimes \mathcal{H}_{N}$, there exists a tensor positive source
operator $T_{S_{1}\times \cdots \times \underset{\overset{\uparrow }{n}}{1}%
\times \cdots \times S_{N}}^{(\rho )}$ for some $n=1,...,N$, then, under
generalized $N$-partite joint quantum measurements, this state is $%
S_{1}\times \cdots \times \widetilde{S_{n}}\times \cdots \times S_{N}$%
-setting Bell local for an arbitrary number $\widetilde{S_{n}}$ of settings
at $n$-th site.
\end{lemma}

This statement, introduced in \cite{6} by Proposition 5, generalizes
Theorems 1, 2 in \cite{8} to a multipartite case, also, Theorem 2 in \cite{7}
formulated for a bipartite case and symmetric $(S_{1},S_{2})$
quasi-extensions.

Note that a symmetric $(S_{1},S_{2})$ quasi-extension, introduced for a
bipartite state in \cite{7}, constitutes an $S_{1}\times S_{2}$-setting
source operator of a particular type -- tensor positive and symmetric. In
this connection, we stress once again that a source operator $T_{S_{1}\times
\cdots \times S_{N}}^{(\rho )}$ \emph{does exist} for every $N$-partite
state $\rho $ and all integers $S_{1},...,S_{N}\geq 1,$ but it does not need
to be either tensor positive or symmetric (in the sense of \cite{7}).

By Proposition 1, Lemma 1 and condition (\ref{16}), we come to the following
general theorem.

\begin{theorem}
If, for a state $\rho ,$ a tensor positive source operator $T_{S_{1}\times
\cdots \times \underset{\overset{\uparrow }{n}}{1}\times \cdots \times
S_{N}}^{(\rho )}$ for an arbitrary $n=1,...,N,$ exists for any integers $%
S_{1},\ldots ,S_{N}\geq 1,$ then the maximal violation by this state of all
general Bell inequalities is equal to one: $\mathrm{\Upsilon }_{\rho }=1$,
so that, under all generalized $N$-partite joint quantum measurements, this $%
N$-partite quantum state $\rho $ is fully Bell local.
\end{theorem}

For a fully separable $N$-partite quantum state $\rho _{sep}$, tensor
positive $S_{1}\times \cdots \times S_{N}$-setting source operators exist 
\cite{8, 6} for all integers $S_{1},\ldots ,S_{N}\geq 1.$ However, a fully
nonseparable $N$-partite state can also \cite{8, 6} have tensor positive $%
S_{1}\times \cdots \times S_{N}$-setting source operators.

\section{New bounds}

In this section, we apply Theorem 1 for finding values of $\beta $ for which
a noisy $N$-qudit state (\ref{1}) is fully Bell local under all generalized $%
N$-partite joint quantum measurements.

We stress that our mathematical techniques is valid for all $N\geq 2$.
However, in this article, we do not intend to search for a tensor positive
source operator that, in view of Theorem 1, could reproduce or improve the
known bounds (see in \cite{21, 23}) for full Bell locality of a noisy
two-qudit state (\ref{1}). Our main aim is to find general bounds on full
Bell locality of a noisy $N$-qudit state (\ref{1}) which are valid for all $%
d\geq 2,$ $N\geq 3$ and to study their asymptotics for large $N$ and $d$. As
we discuss this in Introduction, for $d\geq 2,N\geq 3,$ bounds on full Bell
locality of a noisy $N$-qudit state (\ref{1}) are known only within its full
separability.

\subsection{The N-qudit GHZ state}

Let $\left \{ e_{m},m=1,...,d\right \} $ be an orthonormal base in $\mathbb{C%
}^{d}$ and 
\begin{equation}
\rho _{d,N}^{(ghz)}=\frac{1}{d}\sum_{j,j_{1}}\left( |e_{j}\rangle \langle
e_{j_{1}}|\right) ^{\otimes N},\text{ \ \ }d\geq 2,\text{ }N\geq 3,
\label{24}
\end{equation}%
be the $N$-qudit GHZ\ state on $\left( \mathbb{C}^{d}\right) ^{\otimes N}$.
Consider values of a parameter $0\leq \beta \leq 1,$ for which the $N$-qudit
GHZ\ state mixed with white noise:%
\begin{equation}
\beta \rho _{d,N}^{(ghz)}+(1-\beta )\frac{\mathbb{I}_{d}^{\otimes N}}{d^{N}}
\label{25}
\end{equation}%
has a tensor positive $1\times S_{2}\times \cdots \times S_{N}$-setting
source operator for all integers $S_{2},...,S_{N}\geq 1$ and is, therefore, 
\emph{fully} \emph{Bell} \emph{local} by Theorem 1.

Introduce on the complex Hilbert space%
\begin{equation}
\mathbb{C}^{d}\otimes \left( \mathbb{C}^{d}\right) ^{\otimes S_{2}}\otimes
\cdots \otimes \left( \mathbb{C}^{d}\right) ^{\otimes S_{N}}  \label{26}
\end{equation}%
the self-adjoint operator%
\begin{equation}
T_{1\times S_{2}\times \cdots \otimes S_{N}}^{(ghz)}=\frac{1}{d}\sum_{j,%
\text{ }j_{1}}|e_{j}\rangle \langle e_{j_{1}}|\otimes
W_{jj_{1}}^{(d,S_{2})}\otimes \cdots \otimes W_{jj_{1}}^{(d,S_{N})},
\label{27}
\end{equation}%
where 
\begin{align}
W_{jj}^{(d,S_{n})}& =\left( |e_{j}\rangle \langle e_{j}|\right) ^{\otimes
S_{n}},  \label{28} \\
2W_{jj_{1}}^{(d,S_{n})}|_{j\neq j_{1}}& =\frac{\left(
|e_{j}+e_{j_{1}}\rangle \langle e_{j}+e_{j_{1}}|\right) ^{\otimes S_{n}}}{%
2^{S_{n}}}-\frac{\left( |e_{j}-e_{j_{1}}\rangle \langle
e_{j}-e_{j_{1}}|\right) ^{\otimes S_{n}}}{2^{S_{n}}}  \notag \\
& +i\frac{\left( |e_{j}+ie_{j_{1}}\rangle \langle e_{j}+ie_{j_{1}}|\right)
^{\otimes S_{n}}}{2^{S_{n}}}-i\frac{\left( |e_{j}-ie_{j_{1}}\rangle \langle
e_{j}-ie_{j_{1}}|\right) ^{\otimes S_{n}}}{2^{S_{n}}}  \notag
\end{align}%
are operators on $(\mathbb{C}^{d})^{\otimes S_{n}},$ which are invariant
with respect to permutations of spaces $\mathbb{C}^{d}$ in $(\mathbb{C}%
^{d})^{\otimes S_{n}}$ and satisfy the relations%
\begin{equation}
\left( W_{jj_{1}}^{(d,S_{n})}\right) ^{\ast }=W_{j_{1}j}^{(d,S_{n})},\text{
\ \ \ \ \ }\mathrm{tr}_{(\mathbb{C}^{d})^{\otimes (S_{n}-1)}}\left[
W_{jj_{1}}^{(d,S_{n})}\right] =|e_{j}\rangle \langle e_{j_{1}}|.  \label{29}
\end{equation}%
It is easy to verify%
\begin{equation}
\mathrm{tr}_{(\mathbb{C}^{d})^{\otimes (S_{2}-1)}\otimes \cdots \otimes (%
\mathbb{C}^{d})^{\otimes (S_{N}-1)}}\left[ T_{1\times S_{2}\times \cdots
\otimes S_{N}}^{(ghz)}\right] =\rho _{d,N}^{(ghz)},  \label{30}
\end{equation}%
so that, by definition (\ref{18}), the self-adjoint operator (\ref{27})
constitutes a $1\times S_{2}\times \cdots \times S_{N}$-setting source
operator for the $N$-qudit GHZ\ state $\rho _{d,N}^{(ghz)}$.

On space (\ref{26}) introduce also the positive operator%
\begin{equation}
T_{1\times S_{2}\times \cdots \times S_{N}}^{(1)}=C\sum_{j\neq j_{1}}\mathbb{%
I}_{d}\otimes \widetilde{W}_{jj_{1}}^{(d,S_{2})}\otimes \sum_{l>l_{1}}%
\widetilde{W}_{ll_{1}}^{(d,S_{3})}\otimes \cdots \otimes \sum_{k>k_{1}}%
\widetilde{W}_{kk_{1}}^{(d,S_{N})}  \label{31}
\end{equation}%
with a constant $C>0$ and positive operators%
\begin{eqnarray}
\widetilde{W}_{jj}^{(d,S_{2})} &=&\left( |e_{j}\rangle \langle e_{j}|\right)
^{\otimes S_{2}},  \label{32} \\
2\widetilde{W}_{jj_{1}}^{(d,S_{n})}|_{j\neq j_{1}} &=&\frac{\left(
|e_{j}+e_{j_{1}}\rangle \langle e_{j}+e_{j_{1}}|\right) ^{\otimes S_{n}}}{%
2^{S_{n}}}+\frac{\left( |e_{j}-e_{j_{1}}\rangle \langle
e_{j}-e_{j_{1}}|\right) ^{\otimes S_{n}}}{2^{S_{n}}}  \notag \\
&&+\frac{\left( |e_{j}+ie_{j_{1}}\rangle \langle e_{j}+ie_{j_{1}}|\right)
^{\otimes S_{n}}}{2^{S_{n}}}+\frac{\left( |e_{j}-ie_{j_{1}}\rangle \langle
e_{j}-ie_{j_{1}}|\right) ^{\otimes S_{n}}}{2^{S_{n}}}  \notag
\end{eqnarray}%
on $(\mathbb{C}^{d})^{\otimes S_{n}},$ invariant with respect to
permutations of spaces $\mathbb{C}^{d}$ in $(\mathbb{C}^{d})^{\otimes S_{n}}$
and satisfying the relations 
\begin{equation}
\widetilde{W}_{jj_{1}}^{(d,S_{n})}=\widetilde{W}_{j_{1}j}^{(d,S_{n})},\text{
\ \ \ }\mathrm{tr}_{(\mathbb{C}^{d})^{S_{n}-1}}\left[ \widetilde{W}%
_{jj_{1}}^{(d,S_{n})}\right] =\delta _{jj_{1}}|e_{j}\rangle \langle
e_{j}|+(1-\delta _{jj_{1}})\left( |e_{j}\rangle \langle e_{j}|\text{ }+\text{
}|e_{j_{1}}\rangle \langle e_{j_{1}}|\right) .  \label{33}
\end{equation}%
Note that, for operators (\ref{28}), (\ref{32}), the relation%
\begin{equation}
\left \vert \mathrm{tr}\left[ X\text{ }W_{jj_{1}}^{(d,S_{n})}\right] \right
\vert \leq \mathrm{tr}\left[ X\text{ }\widetilde{W}_{jj_{1}}^{(d,S_{n})}%
\right]  \label{34}
\end{equation}%
holds for all positive operators $X$ on $(\mathbb{C}^{d})^{\otimes S_{n}}$
and all $j,j_{1}=1,...,d.$

In view of (\ref{32}), (\ref{33}), we have 
\begin{eqnarray}
\sum_{l>l_{1}}\mathrm{tr}_{(\mathbb{C}^{d})^{\otimes (S_{n}-1)}}\widetilde{W}%
_{ll_{1}}^{(d,S_{n})} &=&(d-1)\mathbb{I}_{d},  \label{35} \\
\sum_{l\neq l_{1}}\mathrm{tr}_{(\mathbb{C}^{d})^{\otimes (S_{n}-1)}}%
\widetilde{W}_{ll_{1}}^{(d,S_{n})} &=&2(d-1)\mathbb{I}_{d}.  \notag
\end{eqnarray}%
This implies 
\begin{eqnarray}
&&\mathrm{tr}_{(\mathbb{C}^{d})^{\otimes (S_{2}-1)}\otimes \cdots \otimes (%
\mathbb{C}^{d})^{\otimes (S_{N}-1)}}\left[ T_{1\times S_{2}\times \cdots
\times S_{N}}^{(1)}\right]  \label{36} \\
&=&2C(d-1)^{N-1}\mathbb{I}_{d}^{\otimes N},  \notag
\end{eqnarray}%
so that if%
\begin{equation}
C=C_{d,N}:=\frac{1}{2d^{N}(d-1)^{N-1}},  \label{37}
\end{equation}%
then, by definition (\ref{18}), the self-adjoint operator $T_{1\times
S_{2}\times \cdots \times S_{N}}^{(1)}$ is a $1\times S_{2}\times \cdots
\times S_{N}$-setting source operator for the maximally mixed state $\mathbb{%
I}_{d}^{\otimes N}/d^{N}$.

From relations (\ref{30}), (\ref{36}) it follows that, for all integers $%
S_{2},...,S_{N}\geq 1,$ the self-adjoint operator 
\begin{equation}
\beta T_{1\times S_{2}\times \cdots \otimes S_{N}}^{(ghz)}+(1-\beta
)T_{1\times S_{2}\times \cdots \times S_{N}}^{(1)},\text{ \ \ \ }0\leq \beta
\leq 1,  \label{38}
\end{equation}%
constitutes a $1\times S_{2}\times \cdots \times S_{N}$-setting source
operator for a noisy GHZ\ state (\ref{25}).

In Lemma 2 of Appendix, we find the range of $\beta ,$ for which the $%
1\times S_{2}\times \cdots \times S_{N}$-setting source operator (\ref{38})
is tensor positive. This range does not depend on integers $%
S_{2},...,S_{N}\geq 1,$ so that, by Lemma 1 and Theorem 1, we have the
following new result.

\begin{proposition}
Under generalized $N$-partite joint quantum measurements, the $N$-qudit GHZ
state $\rho _{d,N}^{(ghz)},$ $d\geq 2,N\geq 3,$ mixed with white noise is
fully Bell local for all%
\begin{equation}
\beta \leq \beta _{loc}^{(ghz,d,N)}=\frac{1}{1+2d^{N-1}(d-1)^{N-1}}.
\label{39}
\end{equation}
\end{proposition}

For $N=2$, the full locality bound (\ref{39}) is, of course, also true but
it falls into\emph{\ }the known range \cite{21, 23} for separability of the
two-qudit GHZ state mixed with white noise and is not, therefore,
interesting.

For large $N$ and $d$ , asymptotics of this new bound have the forms: 
\begin{equation}
\beta _{loc}^{(ghz,d,N)}\underset{N\gg 1}{\backsimeq }\frac{1}{%
2d^{N-1}(d-1)^{N-1}},\text{ \ \ \ \ \ \ }\beta _{loc}^{(ghz,d,N)}\underset{%
d\gg 1}{\backsimeq }\frac{1}{2d^{2N-2}}.  \label{40'}
\end{equation}

\subsection{Arbitrary nonlocal N-qudit state}

Let us now find a bound on full Bell locality of a noisy $N$-qudit state (%
\ref{1}) for an arbitrary state $\rho _{d,N}.$ We first analyze a bound for
a pure state $|\psi _{d,N}\rangle \langle \psi _{d,N}|$, $d\geq 2,N\geq 3,$
and further by convexity extend the derived result to an arbitrary $\rho
_{d,N}$.

Every pure state on $\left( \mathbb{C}^{d}\right) ^{\otimes N}$ admits the
decomposition

\begin{equation}
|\psi _{d,N}\rangle \langle \psi _{d,N}|\text{ }=\sum \varsigma
_{mj...k}\varsigma _{m_{1}j_{1}...k_{1}}^{\ast }|e_{m}\rangle \langle
e_{m_{1}}|\otimes |e_{j}\rangle \langle e_{j_{1}}|\otimes \cdots \otimes
|e_{k}\rangle \langle e_{k_{1}}|  \label{41}
\end{equation}%
where $\sum_{m,j,...,k}\left \vert \varsigma _{mj...k}\right \vert ^{2}=1.$
By introducing the normalized vectors%
\begin{align}
\phi _{j...k}& =\frac{1}{\alpha _{j...k}}\sum_{m}\varsigma _{mj...k}e_{m},\
\ \ \ \ \ \ \ \left \Vert \phi _{j...k}\right \Vert =1,  \label{42} \\
\alpha _{j...k}& =\left( \dsum \limits_{m}|\varsigma _{mj...k}|^{2}\right)
^{1/2},\text{ \ \ }\sum_{\underset{N-1}{\underbrace{j,...,k}}}(\alpha
_{j...k})^{2}=1,  \notag
\end{align}%
we rewrite decomposition (\ref{41}) in the form%
\begin{equation}
|\psi _{d,N}\rangle \langle \psi _{d,N}|\text{ }=\sum \alpha _{j...k}\alpha
_{j_{1}...k}|\phi _{j...k}\rangle \langle \phi _{j_{1}...k_{1}}|\otimes
|e_{j}\rangle \langle e_{j_{1}}|\otimes \cdots \otimes |e_{k}\rangle \langle
e_{k_{1}}|  \label{43}
\end{equation}%
where all coefficients $\alpha _{j...k}$ are nonnegative.

In view of this decomposition, introduce on the Hilbert space (\ref{26}) the
self-adjoint operator%
\begin{equation}
T_{1\times S_{2}\times \cdots \otimes S_{N}}^{(\psi
_{d,N})}=\sum_{j,...,k}\alpha _{j...k}\alpha _{j_{1}...k_{1}}|\phi
_{j...k}\rangle \langle \phi _{j_{1}...k_{1}}|\otimes
W_{jj_{1}}^{(d,S_{2})}\otimes \cdots \otimes W_{kk_{1}}^{(d,S_{N})}
\label{44}
\end{equation}%
where operators $W_{ll_{1}}^{(d,S_{n})}$ are defined by (\ref{28}). It is
easy to verify%
\begin{equation}
\mathrm{tr}_{(\mathbb{C}^{d})^{\otimes (S_{2}-1)}\otimes \cdots \otimes (%
\mathbb{C}^{d})^{\otimes (S_{N}-1)}}\left[ T_{1\times S_{2}\times \cdots
\times S_{N}}^{(\psi _{d,N})}\right] =|\psi _{d,N}\rangle \langle \psi
_{d,N}|,  \label{45}
\end{equation}%
so that, by definition (\ref{18}) the self-adjoint operator $T_{1\times
S_{2}\times \cdots \times S_{N}}^{(\psi _{d,N})}$ constitutes a $1\times
S_{2}\times \cdots \times S_{N}$-setting source operator for a pure state $%
|\psi _{d,N}\rangle \langle \psi _{d,N}|$.

On the space (\ref{26}) consider also the positive operator%
\begin{equation}
T_{1\times S_{2}\times \cdots \times S_{N}}^{(2)}=\widetilde{C}%
\sum_{(j,...,k)\neq (j_{1},...,k_{1})}\mathbb{I}_{d}\otimes \widetilde{W}%
_{jj_{1}}^{(d,S_{2})}\otimes \cdots \otimes \widetilde{W}%
_{kk_{1}}^{(d,S_{N})}  \label{46}
\end{equation}%
with a constant $\widetilde{C}>0$ and positive operators $\widetilde{W}%
_{ll_{1}}^{(d,S_{2})}$ on $(\mathbb{C}^{d})^{\otimes S_{n}}$ defined by (\ref%
{32}). Taking into account that $\sum_{j}\mathrm{tr}_{(\mathbb{C}%
^{d})^{\otimes (S_{n}-1)}}\widetilde{W}_{jj}^{(d,S_{n})}=\mathbb{I}_{d}$ and 
\begin{equation}
\sum_{j,\text{ }j_{1}}\mathrm{tr}_{(\mathbb{C}^{d})^{\otimes (S_{n}-1)}}%
\widetilde{W}_{jj_{1}}^{(d,S_{n})}=(2d-1)\mathbb{I}_{d},  \label{47}
\end{equation}%
we derive 
\begin{eqnarray}
&&\mathrm{tr}_{(\mathbb{C}^{d})^{\otimes (S_{2}-1)}\otimes \cdots \otimes (%
\mathbb{C}^{d})^{\otimes (S_{N}-1)}}\left[ T_{1\times S_{2}\times \cdots
\times S_{N}}^{(2)}\right]  \label{48} \\
&=&\widetilde{C}\left \{ (2d-1)^{N-1}-1\right \} \text{ }\mathbb{I}%
_{d}^{\otimes N}.  \notag
\end{eqnarray}%
Hence, if 
\begin{equation}
\widetilde{C}=\widetilde{C}_{d,N}:=\frac{1}{d^{N}\left \{
(2d-1)^{N-1}-1\right \} }  \label{49}
\end{equation}%
then, by (\ref{18}), the operator $T_{1\times S_{2}\times \cdots \times
S_{N}}^{(2)}$ constitutes a $1\times S_{2}\times \cdots \times S_{N}$%
-setting source operator for the maximally mixed state $\mathbb{I}%
_{d}^{\otimes N}/d^{N}.$

From relations (\ref{45}), (\ref{48}) it follows that the self-adjoint
operator 
\begin{equation}
\beta T_{1\times S_{2}\times \cdots \times S_{N}}^{(\psi _{d,N})}+(1-\beta
)T_{1\times S_{2}\times \cdots \times S_{N}}^{(2)}  \label{50}
\end{equation}%
constitutes a $1\times S_{2}\times \cdots \times S_{N}$-setting source
operator for a mixture (\ref{1}) of a pure state $|\psi _{d,N}\rangle
\langle \psi _{d,N}|$ with white noise.

In Lemma 3 of Appendix, we find a range of $\beta ,$ for which this source
operator is tensor positive. This range does not depend on integers $%
S_{2},...,S_{N}\geq 1,$ so that, by Lemma 3 and Theorem 1, an arbitrary pure
state $|\psi _{d,N}\rangle \langle \psi _{d,N}|,$ $d\geq 2,$ $N\geq 2,$
mixed with white noise is fully Bell local for all%
\begin{eqnarray}
\beta &\leq &\beta _{loc}^{(\psi _{d,N})}=\frac{1}{1+d^{N}\left \{
(2d-1)^{N-1}-1\right \} \gamma _{\psi _{d,N}}^{\max }},  \label{51} \\
\gamma _{\psi _{d,N}}^{\max } &=&\max_{j,...,k}\alpha _{j...k}^{2}.  \notag
\end{eqnarray}

Taking further into account relation $\gamma _{\psi _{d,N}}^{\max }\geq 
\frac{1}{d^{N-1}}$ valid for all pure states $|\psi _{d,N}\rangle \langle
\psi _{d,N}|$ and that, for each mixed state $\rho _{d,N}=\sum_{j}\xi
_{j}|\psi _{d,N}^{(j)}\rangle \langle \psi _{d,N}^{(j)}|,$ $\xi _{j}>0,$ $%
\sum_{j}\xi _{j}=1,$ the sum $\sum_{j}\xi _{j}T_{1\times S_{2}\times \cdots
\times S_{N}}^{(\psi _{d,N}^{(j)})}$ is a $1\times S_{2}\times \cdots \times
S_{N}$-setting source operator for state $\rho _{d,N},$ we come by Theorem 1
to the following new result.

\begin{proposition}
Under generalized N-partite joint quantum measurements, an arbitrary $N$%
-qudit state $\rho _{d,N},$ $d\geq 2,$ $N\geq 3,$ mixed with white noise is
fully Bell local for all $\beta \leq \beta _{loc}^{(\rho _{d,N})}$ where%
\begin{eqnarray}
\frac{1}{d^{N}(2d-1)^{N-1}-d^{N}+1} &\leq &\beta _{loc}^{(\rho _{d,N})}\leq 
\frac{1}{d(2d-1)^{N-1}-d+1}.  \label{52} \\
&&  \notag
\end{eqnarray}
\end{proposition}

For $N=2$, the full locality bound (\ref{52}) is also true but it falls into%
\emph{\ }the known range \cite{21} for separability of an arbitrary
two-qudit state $\rho _{d,2}$ mixed with white noise and is not, therefore,
interesting.

For large $N$ and $d$ , asymptotics of this new bound have the forms:%
\begin{eqnarray}
&&\frac{1}{d^{N}(2d-1)^{N-1}}\underset{_{N\gg 1}}{\lesssim }\beta
_{loc}^{(\rho _{d,N})}\underset{_{N\gg 1}}{\lesssim }\frac{1}{d(2d-1)^{N-1}},
\label{53} \\
&&\frac{1}{2^{N-1}d^{2N-1}}\underset{_{d\gg 1}}{\lesssim }\beta
_{loc}^{(\rho _{d,N})}\underset{_{d\gg 1}}{\lesssim }\frac{1}{2^{N-1}d^{N}}.
\notag
\end{eqnarray}

\section{Discussion}

In the present paper, we have presented Theorem 1, specifying the sufficient
condition for full Bell locality of an $N$-partite quantum state via its
dilation characteristics, and, due to this condition, we have derived for
all $d\geq 2$ and all $N\geq 3$ a new bound (\ref{39}) on full Bell locality
of the $N$-qudit GHZ\ state mixed with white noise and a new bound (\ref{52}%
) for full Bell locality of an arbitrary $N$-qudit state mixed with white
noise.

As we discuss this in Introduction, to our knowledge, for arbitrary $d\geq
2, $ $N\geq 3,$ bounds in $d,$ $N$ on full Bell locality of a noisy $N$%
-qudit state (\ref{1}) have been known in the literature only within its
full separability.

Let us now compare our new full Bell locality bounds (\ref{39}), (\ref{52})
with full separability bounds (\ref{1.1})--(\ref{1.4}) known for a noisy $N$%
-qudit state (\ref{1}) with $N\geq 3.$

For a prime $d\geq 2$ and an arbitrary $N\geq 3,$ the full Bell locality
bound $\beta _{loc}^{(ghz,d,N)}$ in (\ref{39}) for a noisy $N$-qudit GHZ
state (\ref{25}) falls into the range for its full separability: $\beta \leq
\beta _{sep}^{(ghz,d,N)}|_{prime\text{ }d}=\frac{1}{1+d^{N-1}},$ and is not,
therefore, interesting.

However, comparing the full Bell locality bound (\ref{39}) for a noisy $N$%
-qudit GHZ state (\ref{25}) with the lower bound in (\ref{1.4}) on its full
separability, we have 
\begin{equation}
\beta _{loc}^{(ghz,d,N)}=\frac{1}{1+2d^{N-1}(d-1)^{N-1}}>\frac{1}{1+d^{2N-1}}%
,\text{ \ \ }\forall d\geq 2,N\geq 3.  \label{54}
\end{equation}%
This means that, for a non-prime $d>3$ and an arbitrary\emph{\ }$N\geq 3,$ a
noisy $N$-qudit GHZ state (\ref{25}) is fully Bell local for all $\beta \leq 
\frac{1}{1+2d^{N-1}(d-1)^{N-1}},$ whereas it is definitely known to be fully
separable if $\beta \leq \frac{1}{1+d^{2N-1}}$. This new result on full Bell
locality of a noisy $N$-qudit GHZ state (\ref{25}) for a non-prime $d>3$ and
an arbitrary $N\geq 3$ does not, however, specify either in the interval 
\begin{equation}
\frac{1}{1+d^{2N-1}}<\beta \leq \frac{1}{1+2d^{N-1}(d-1)^{N-1}},  \label{55}
\end{equation}%
this noisy state is fully separable or fully nonseparable. As we discuss in
Introduction, for a non-prime $d>3$ and an arbitrary $N\geq 3,$ it is only
known \cite{17.2, 17.1, 17.3} that a noisy $N$-qudit GHZ state\textbf{\ }(%
\ref{25}) is fully nonseparable for all $\beta >\frac{1}{1+d^{N-1}}$\textbf{.%
}

For an arbitrary $N$-qudit state $\rho _{_{d,N}}$ mixed with white noise,
the lower bound in (\ref{52})\ is within the known full separability range
in (\ref{1.3}) while the upper bound in (\ref{52}) is essentially out of
this full separability range for all $d\geq 2,N\geq 3.$ This means that, for
some $N$-qudit state $\rho _{_{d,N}}$ mixed with white noise, a possible gap
between the bound in (\ref{52}) on its full Bell locality and the known
bound (\ref{1.3}) on its full separability can reach the value 
\begin{equation}
\Delta _{\rho _{_{d,N}}}^{\max }=\frac{1}{d(2d-1)^{N-1}-d+1}-\frac{1}{%
d^{2N-1}+1},\text{ \ \ \ }d\geq 2,\text{ }N\geq 3.  \label{58}
\end{equation}%
For example, for $N=3,$ $d=2$, this gap is equal to $0.94\beta _{sep}^{(\rho
_{2,3})}$. Therefore, for some three-qudit state $\rho _{2,3},$ the full
Bell locality bound $\beta _{loc}^{(\rho _{2,3})}$ in (\ref{52}) can be
almost twice more than the known full separability bound $\beta
_{sep}^{(\rho _{2,3})}$ in (\ref{1.1}).\smallskip

We note that, in section 4, our choices (\ref{31}), (\ref{46}) of $1\times
S_{2}\times \cdots \times S_{N}$-setting source operators for the maximally
mixed state $\mathbb{I}_{d}^{\otimes N}/d^{N}$ are definitely not optimal,
the same concerns our evaluation of tensor positivity in Lemmas 2, 3. This
allows us to believe that the derived full Bell locality bounds can be
further considerably improved.

In conclusion, we have derived new general bounds, expressed in terms of $%
d,N $ and valid for all $d\geq 2$ and all $N\geq 3,$ on full Bell locality
under generalized quantum measurements of (i) the $N$-qudit GHZ state mixed
with white noise and (ii) an arbitrary $N$-qudit state mixed with white
noise. The new full locality bounds are beyond the known ranges for full
separability of these noisy states.

\section{Appendix}

\begin{lemma}
For arbitrary $d\geq 2,$ $N\geq 2,$ the source operator (\ref{38}) on space (%
\ref{26}) is tensor positive for all 
\begin{equation}
\beta \leq \frac{1}{1+2d^{N-1}(d-1)^{N-1}}.  \tag{A1}
\end{equation}
\end{lemma}

\begin{proof}
In view of (\ref{22}) and the structure of operators $T_{1\times S_{2}\times
\cdots \times S_{N}}^{(ghz)}$ and $T_{1\times S_{2}\times \cdots \times
S_{N}}^{(1)}$, given by relations (\ref{27}), (\ref{31}), for finding a
range of tensor positivity of the source operator (\ref{38}), we need to
find $\beta $ for which the expression 
\begin{eqnarray}
&&(1-\beta )\text{ }\mathrm{tr}\left[ T_{1\times S_{2}\times \cdots \times
S_{N}}^{(1)}\left( X_{1}\otimes X_{S_{2}}\otimes \cdots \otimes
X_{S_{N}}\right) \right]  \TCItag{A2} \\
&&+\beta \text{ }\mathrm{tr}\left[ T_{1\times S_{2}\times \cdots \times
S_{N}}^{(ghz)}\left( X_{1}\otimes X_{S_{2}}\otimes \cdots \otimes
X_{S_{N}}\right) \right]  \notag
\end{eqnarray}%
is nonnegative for all positive operators $X_{1}$ on $\mathbb{C}^{d}$ and $%
X_{S_{n}}$ on $(\mathbb{C}^{d})^{\otimes S_{n}},$ $n=1,...,N$.

Moreover, since in decomposition (\ref{27}) the term with $j=j_{1}$ is
positive, it is suffice to evaluate nonnegativity of 
\begin{eqnarray}
\Delta &=&(1-\beta )C_{d,N}\sum_{j\neq j_{1}}\mathrm{tr}[X_{1}]\text{ }%
\mathrm{tr}[\widetilde{W}_{jj_{1}}^{(d,S_{2})}X_{S_{2}}]\sum_{l>l_{1}}%
\mathrm{tr}[\widetilde{W}_{ll_{1}}^{(d,S_{3})}X_{S_{3}}]\cdot \ldots \cdot
\sum_{k>k_{1}}\mathrm{tr}[\widetilde{W}_{kk_{1}}^{(d,S_{N})}X_{S_{N}}] 
\notag \\
&&+\frac{\beta }{d}\sum_{j\neq j_{1}}\langle e_{j}|X_{1}|e_{j_{1}}\rangle 
\text{ }\mathrm{tr}[W_{jj_{1}}^{(d,S_{2})}X_{S_{2}}]\cdot \ldots \cdot 
\mathrm{tr}[W_{jj_{1}}^{(d,S_{N})}X_{N,S_{N}}].  \TCItag{A3}
\end{eqnarray}%
Taking into account (\ref{34}), relations $\left \vert \langle
e_{j}|X_{1}|e_{j_{1}}\rangle \text{ }\right \vert \leq \mathrm{tr}[X_{1}]$
and $\widetilde{W}_{jj_{1}}^{(d,S_{2})}=\widetilde{W}_{j_{1}j}^{(d,S_{2})},$
we have 
\begin{eqnarray}
&&\sum_{j\neq j_{1}}\mathrm{tr}[X_{1}]\text{ }\mathrm{tr}[\widetilde{W}%
_{jj_{1}}^{(d,S_{2})}X_{S_{2}}]\cdot \sum_{l>l_{1}}\mathrm{tr}[\widetilde{W}%
_{ll_{1}}^{(d,S_{N})}X_{S_{3}}]\cdot \ldots \cdot \sum_{k>k_{1}}\mathrm{tr}[%
\widetilde{W}_{kk_{1}}^{(d,S_{N})}X_{S_{N}}]  \TCItag{A4} \\
&\geq &\sum_{j\neq j_{1}}\mathrm{tr}[X_{1}]\text{ }\mathrm{tr[}\widetilde{W}%
_{jj_{1}}^{(d,S_{2})}X_{S_{2}}]\text{ }\mathrm{tr}[\widetilde{W}%
_{jj_{1}}^{(d,S_{3})}X_{S_{3}}]\cdot \ldots \cdot \mathrm{tr}[\widetilde{W}%
_{jj_{1}}^{(N)}X_{S_{N}}]  \notag \\
&\geq &\sum_{j\neq j_{1}}\left \vert \langle e_{j}|X_{1}|e_{j_{1}}\rangle
\right \vert \left \vert \mathrm{tr[}W_{jj_{1}}^{(d,S_{2})}X_{S_{2}}]\right%
\vert \cdot \ldots \cdot \left \vert \mathrm{tr}[W_{jj_{1}}^{(N)}X_{S_{N}}]%
\right \vert .  \notag
\end{eqnarray}%
From (A3), (A4) it follows%
\begin{equation}
\Delta \geq \left \{ (1-\beta )C-\frac{1}{d}\beta \right \} \sum_{j\neq
j_{1}}\mathrm{tr}[X_{1}]\mathrm{tr[}\widetilde{W}%
_{jj_{1}}^{(d,S_{2})}X_{S_{2}}]\cdot \ldots \cdot \mathrm{tr}[\widetilde{W}%
_{jj_{1}}^{(N)}X_{S_{N}}].  \tag{A5}
\end{equation}%
for all $X_{1}\geq 0$ on $\mathbb{C}^{d}$ and $X_{S_{n}}\geq 0$ on $(\mathbb{%
C}^{d})^{\otimes S_{n}}.$ Recall that, due to their definition (\ref{32}),
all operators $\widetilde{W}_{jj_{1}}^{(d,S_{n})}$ are positive. Therefore,
if 
\begin{equation}
(1-\beta )C_{d,N}-\frac{1}{d}\beta \geq 0\text{ }\Leftrightarrow \text{ }%
\beta \leq \frac{dC_{d,N}}{1+dC_{d,N}},  \tag{A6}
\end{equation}%
then $\Delta \geq 0\ $and the source operator (\ref{38}) is tensor positive.
In view of (\ref{37}), this proves the statement.
\end{proof}

\begin{lemma}
For arbitrary $d\geq 2,$ $N\geq 2,$ the source operator (\ref{50}) is tensor
positive for all 
\begin{eqnarray}
\beta &\leq &\beta _{loc}^{(\psi _{d,N})}=\frac{1}{1+d^{N}\left \{ \left(
2d-1\right) ^{N-1}-1\right \} \gamma _{\psi _{d,N}}^{\max }},  \TCItag{A7} \\
\gamma _{\psi _{d,N}}^{\max } &=&\max_{j,...,k}\alpha _{j...k}^{2}.  \notag
\end{eqnarray}
\end{lemma}

\begin{proof}
Quite similarly to our proof in Lemma 2, let us analyse nonnegativity of the
expression 
\begin{eqnarray}
&&(1-\beta )\text{ }\mathrm{tr}\left[ T_{1\times S_{2}\times \cdots \times
S_{N}}^{(2)}\left( X_{1}\otimes X_{S_{2}}\otimes \cdots \otimes
X_{S_{N}}\right) \right]  \TCItag{A8} \\
&&+\beta \text{ }\mathrm{tr}\left[ T_{1\times S_{2}\times \cdots \times
S_{N}}^{(\psi _{d,N})}\left( X_{1}\otimes X_{S_{2}}\otimes \cdots \otimes
X_{S_{N}}\right) \right]  \notag
\end{eqnarray}%
for all positive $X_{1}$ on $\mathbb{C}^{d}$ and $X_{n,S_{n}}$ on $(\mathbb{C%
}^{d})^{\otimes S_{n}},$ $n=1,...,N$.

Since in decomposition (\ref{44}) the term with $(j,...,k)=(j_{1},...,k_{1})$
is positive, it is suffice to evaluate 
\begin{eqnarray}
\Delta &=&(1-\beta )\widetilde{C}_{d,N}\sum_{(j,...,k)\neq (j_{1},...,k_{1})}%
\mathrm{tr}[X_{1}]\text{ }\mathrm{tr}[\widetilde{W}%
_{jj_{1}}^{(d,S_{2})}X_{S_{2}}]\cdot \ldots \cdot \mathrm{tr}[\widetilde{W}%
_{kk_{1}}^{(d,S_{N})}X_{S_{N}}]  \TCItag{A9} \\
&&+\beta \sum_{(j,...,k)\neq (j_{1},...,k_{1})}\alpha _{j...k}\alpha
_{j_{1}...k_{1}}\langle \phi _{j_{1}...k_{1}}|X_{1}|\phi _{j...k}\rangle 
\text{ }\mathrm{tr}[W_{jj_{1}}^{(d,S_{2})}X_{S_{2}}]\cdot \ldots \cdot 
\mathrm{tr}[W_{kk_{1}}^{(d,S_{N})}X_{S_{N}}].  \notag
\end{eqnarray}%
Taking into account (\ref{34}) and relation $\left \vert \langle \phi
_{j_{1}...k_{1}}|X_{1}|\phi _{j...k}\rangle \right \vert \leq \mathrm{tr}%
[X_{1}],$ we have%
\begin{eqnarray}
&&\sum_{(j,...,k)\neq (j_{1},...,k_{1})}\mathrm{tr}[X_{1}]\text{ }\mathrm{tr}%
[\widetilde{W}_{jj_{1}}^{(d,S_{2})}X_{S_{2}}]\cdot \ldots \cdot \mathrm{tr}[%
\widetilde{W}_{kk_{1}}^{(d,S_{N})}X_{S_{N}}]  \TCItag{A10} \\
&\geq &\sum_{(j,...,k)\neq (j_{1},...,k_{1})}\left \vert \langle \phi
_{j_{1}...k_{1}}|X_{1}|\phi _{j...k}\rangle \text{ }\right \vert \left \vert 
\mathrm{tr[}W_{jj_{1}}^{(d,S_{2})}X_{S_{2}}]\right \vert \cdot \ldots \cdot
\left \vert \mathrm{tr}[W_{kk_{1}}^{(N)}X_{S_{N}}]\right \vert ,  \notag
\end{eqnarray}%
so that 
\begin{equation}
\Delta \geq \left \{ (1-\beta )\widetilde{C}_{d,N}-\gamma _{\psi
_{d,N}}^{\max }\beta \right \} \sum_{(j,...,k)\neq (j_{1},...,k_{1})}\mathrm{%
tr}[X_{1}]\text{ }\mathrm{tr}[\widetilde{W}_{jj_{1}}^{(d,S_{2})}X_{S_{2}}]%
\cdot \ldots \cdot \mathrm{tr}[\widetilde{W}_{kk_{1}}^{(d,S_{N})}X_{S_{N}}] 
\tag{A11}
\end{equation}%
for all $X_{1}\geq 0$ on $\mathbb{C}^{d}$ and $X_{S_{n}}\geq 0$ on $(\mathbb{%
C}^{d})^{\otimes S_{n}}.$ Recall, that, due to their definition (\ref{32}),
all operators $\widetilde{W}_{jj_{1}}^{(d,S_{n})}$ are positive. Therefore,
if%
\begin{equation}
(1-\beta )\widetilde{C}_{d,N}-\gamma _{\psi _{d,N}}^{\max }\beta \geq 0\text{
}\Leftrightarrow \text{ }\beta \leq \frac{\widetilde{C}_{d,N}}{\gamma _{\psi
_{d,N}}^{\max }+\widetilde{C}_{d,N}},  \tag{A12}
\end{equation}%
then the source operator (\ref{50}) is tensor positive. In view of (\ref{49}%
), this proves the statement.
\end{proof}

\end{document}